\documentclass{amsart}
\usepackage{latexsym,amscd,array}
\usepackage{enumerate,paralist}
\usepackage{exscale}
\usepackage[centertags]{amsmath}
\usepackage{amssymb,stmaryrd}
\usepackage{amsthm}
\usepackage{dsfont}
\usepackage[all]{xy}
\usepackage{lscape}
\usepackage{pdflscape}
\usepackage{color}
\usepackage{graphicx}
\usepackage{hyperref}
\usepackage{mathdots}

\usepackage{lineno}

\usepackage{thmtools}

\newtheorem{thm}{Theorem}[section]

\newcommand{\myendsymbol}{\ensuremath{\diamondsuit}}

\declaretheorem[
  style=definition,
  title=Example,
  qed={$\myendsymbol$},
  refname={example,examples},
  Refname={Example,Examples},
  sharenumber=thm,
]{exa}

\declaretheorem[
  style=definition,
  title=Definition,
  qed={$\myendsymbol$},
  sharenumber=thm,
]{dfn}

\declaretheorem[
  style=definition,
  title=Hypothesis,
  qed={$\myendsymbol$},
  sharenumber=thm,
]{hyp}

\declaretheorem[
  style=definition,
  title=Notation,
  qed={$\myendsymbol$},
  sharenumber=thm,
]{ntn}

\declaretheorem[
  style=definition,
  title=Remark,
  qed={$\myendsymbol$},
  sharenumber=thm,
]{rmk}

\newtheorem{prop}[thm]{Proposition}
\newtheorem{lem}[thm]{Lemma}
\newtheorem{cor}[thm]{Corollary}

\newcommand{\bolda}{{\mathbf{a}}}

\newcommand{\bolde}{{\mathbf{e}}}

\newcommand{\boldu}{{\mathbf{u}}}
\newcommand{\boldv}{{\mathbf{v}}}

\newcommand{\bsx}{{\boldsymbol x}}
\newcommand{\bsy}{{\boldsymbol y}}
\newcommand{\bsdel}{{\boldsymbol \del}}

\newcommand{\rank}{{\operatorname{rk}}}

\newlength{\myl}
\newcommand{\del}{\partial}

\newcommand{\into}{\hookrightarrow}

\renewcommand{\to}{\longrightarrow}

\newcommand{\minus}{\smallsetminus}

\newcommand{\calB}{\mathcal{B}}

\newcommand{\calF}{\mathcal{F}}
\newcommand{\calG}{\mathcal{G}}

\newcommand{\calT}{\mathcal{T}}
\newcommand{\calU}{\mathcal{U}}

\newcommand{\matroidM}{{\mathsf M}}

\newcommand{\frakt}{{\mathfrak{t}}}

\newcommand{\fraky}{{\mathfrak{y}}}

\newcommand{\CC}{\mathbb{C}}

\newcommand{\KK}{\mathbb{K}}

\newcommand{\NN}{\mathbb{N}}

\newcommand{\QQ}{\mathbb{Q}}
\newcommand{\RR}{\mathbb{R}}

\newcommand{\ZZ}{\mathbb{Z}}

\newcommand{\VExt}{{V_{\textrm{Ext}}}}

\DeclareMathOperator{\depth}{\textup{depth}}

\DeclareMathOperator{\Span}{\textup{span}\,}

\DeclareMathOperator{\Supp}{\textup{Supp}}

\newcommand{\Feynman}{{\mathrm{Feyn}}}
\newcommand{\mt}{{\mathrm{m.t.}}}

\numberwithin{equation}{subsection}

\begin{document}
\title{On Feynman graphs, matroids, and GKZ-systems}


\author{Uli Walther}
\address{ Uli~Walther\\
  Purdue University\\
  Dept.\ of Mathematics\\
  150 N.\ University St.\\
  West Lafayette, IN 47907\\ USA}
\email{walther@purdue.edu}

\thanks{UW was supported by NSF grant DMS-2100288 and by Simons
  Foundation Collaboration Grant for Mathematicians \#580839.}

\begin{abstract}
We show in several important cases that the $A$-hypergeometric system
attached to a Feynman diagram in Lee--Pomeransky form, obtained by viewing the
momenta and the nonzero masses as indeterminates, has a normal
underlying semigroup. This continues a quest initiated by Klausen, and
studied by Helmer and Tellander. In the process we identify several
relevant matroids related to the situation and explore their
relationships. 
\end{abstract}

\maketitle

\setcounter{tocdepth}{3}
\tableofcontents

\section{Introduction}
Throughout, $G$ is a graph with edge set $E:=E_G$ and vertex set
$V:=V_G$.\footnote{We will typically use $E$ and reserve $E_G$ for
  cases where extra clarity is needed, for example when several graphs
  are around.}
Denote by $\calT^i_G$ its set of \emph{$i$-forests}, so
$F\in\calT^i_G$ whenever it is circuit-free and the graph on the set
of vertices of $G$ with the set of edges of $F$ has exactly $(i-1)$
more connected components than $G$ does. The nomenclature comes from
the fact that a $i$-forest in a connected graph has $i$ connected
components. If $G$ is connected, a $1$-forest is often called a
\emph{spanning tree}.

In the theory of Feynman integrals, edges correspond to particles, and
vertices to particle interactions.
Some of the vertices are labelled
as ``external''; the set of external vertices is denoted $\VExt$. An external vertex connects to an external
edge (that is not part of $G$) and these external edges represent the
externally measurable in- and output particles that interact according
to the graph.

Throughout we consider a mass function
\[
m\colon E\to \RR_{\geq 0},
\]
and denote by $m_e$ the mass of the particle corresponding to edge
$e$.  As a matter of general notation, we call \emph{massive} the
edges $e$ with $m_e\neq 0$; the other edges are \emph{massless}.

There is a momentum function $p$ on the external vertices of $G$, with
values in the 4-dimensional Minkowski space $\RR^{1,3}$ with
indefinite ``norm''
$p^2=|(p_0,p_1,p_3,p_3)|^2:=p_0^2-(p_1^2+p_2^2+p_3^2)$.  Momentum
conservation dictates that the momenta of the external particles must
sum to zero. We will assume (see Hypothesis \ref{hyp} below) that the
momenta do not satisfy any other constraints. In particular, when
measurements of experiment are taken, the momenta can be seen as
generic (subject to summing to zero); this setup fits most QFTs.

No generality on the Feynman diagram is lost if one assumes that the
underlying graph $G$ be connected, since disconnected graphs describe
separate particle interactions. Slightly more generally, one may assume that the
graph have no \emph{cut vertex}: the removal of any single vertex of $G$
should not increase the number of connected components. This property
is in the Feynman context referred to as (1VI),
short for ``one vertex irreducible''; see for example
\cite{Schultka}. Physically, the presence of a cut vertex means that
the particle interaction can be interpreted as a two-stage process
with independent parts.

A \emph{bridge} is an edge whose removal increases the number of
connected components. In the presence of bridges, as well as when the
graph has edges linking some vertex to itself, the corresponding
Feynman amplitude factors into amplitudes from simpler graphs. In
physics, a connected graph without any edges linking a vertex to
itself, and without bridges is called (1PI), short for ``one particle
irreducible''.  It implies in particular that no edge is part of every
1-forest.

\begin{dfn}
  We will say that the graph $G$ is \emph{strongly 1-irreducible},
  abbreviated as \emph{(s1I)} if it is both one particle irreducible
  and one vertex irreducible. Equivalently, such graphs are connected,
  and have no bridges, cut vertices or edges that link a vertex to itself.
\end{dfn}

Mathematically, the (s1I) property is ``the graphical (or,
equivalently, the co-graphical) matroid to $G$ is connected'', see
Subsection \ref{subsec-matroids} below.

The graph $G$ induces several interesting functions on
\[
\RR^E:=\bigoplus_{e\in E} \RR\cdot \bolde_e,
\]
that lie inside the polynomial ring $\CC[\bsx_E]$ on
variables $\bsx_E:=\{x_e\mid e\in E\}$ indexed by $E$ the \emph{dual
  graph polynomial}
\[
\calU:=\sum_{T\in\calT^1_G} (\bsx^E/\bsx^T),
\]
where here and elsewhere, $\bsx^S:=\prod_{e\in S} x_e$ for any
$S\subseteq E$, and more generally $\bsx^\bolda:=\prod_{e\in E}
x_i^{a_i}$ for $\bolda\in\ZZ^E$.

Given a set of external momenta, a second polynomial can be derived
from $G$, namely
\[
\calF_0:=-\sum_{F\in \calT^2_G} |p(F)|^2(\bsx^E/\bsx^F).
\]
Here, $p(F)$ is the sum of the momenta of the external vertices of $G$
that belong to one of the two components $F$,\footnote{Since the
  momenta sum is zero, both 2-forest components give the same
  coefficient.} compare the introduction of \cite{HT}.

Many QFT techniques take recourse to \emph{Wick rotation}, the
coordinate transformation that multiplies the momentum coordinate
$p_0$ by $\sqrt{-1}$. We shall write $\calF_0^W$ for the result of
Wick rotation on $\calF_0$. The effect is that the Minkowski norm
turns into the Euclidean norm, but it also moves the study of Feynman
amplitudes to the complex domain. For certain purposes, such as
considering families of Feynman type integrals in the spririt
discussed below, this is no actual disadvantage.

%

In contrast to the momenta, there is no
genericity assumption on the masses, and in particular they can be
zero.  One then defines
\[
\calF:=\calU\cdot(\sum_{e\in G}m_e^2x_e)+\calF_0^W.
\]
In the theory of Feynman integrals, in Lee--Pomeransky form, the function
\[
\calG_m:=\calU+\calF= \calU\cdot(1+\sum_{e\in G}m_e^2x_e)+\calF_0^W
\]
and its integrals are relevant, see \cite{Klausen-JHEP20,HT}.

\begin{rmk}\label{rmk-conditions}
  \begin{asparaenum}
    \item The Lee--Pomeransky formalism assumes that the underlying
      graph $G$ is of type (s1I). As noted, if a particle interaction
      is modeled by a graph that is not (s1I) then one can decompose
      the situation into subproblems whose graph is in fact (s1I).
    \item The Lee--Pomeransky form of the Feynman integral assumes
      Wick rotation. This means that one
      must allow for complex components in the momenta, which then
      raises the possibility of cancellation of coefficients in the
      sum $\calF$, resulting in the
      possible disappearance of certain monomials. For
      degree reasons no cancellation can occur between terms of $\calU$
      and terms of $\calF$. 
  \end{asparaenum}
\end{rmk}
In order to avoid the pathologies mentioned the previous
remark, we shall make the following assumptions.
\begin{hyp}[Feynman Hypotheses]\label{hyp}
  Throughout, 
  we shall assume that
  \begin{enumerate} 
  \item the underlying graph $G$ is (s1I) and has at least one edge
    (hence actually at least two);
  \item the values of the momenta are sufficiently generic, so that
    \begin{enumerate}
      \item in the sum $\calU\cdot (\sum_{e\in G}m_e^2x_e)+\calF_0^W$ no
        cancellation of terms occurs, and
      \item no proper subset of $\VExt$
        has zero momentum sum.
    \end{enumerate}    
  \item At least one 2-forest term appears in $\calG_m$.
  \end{enumerate}
\end{hyp}

\begin{rmk}
  \begin{asparaenum}
  \item Hypothesis \ref{hyp}.(1) can be postulated since
    Feynman amplitudes to graphs that fail this condition can be
    decomposed into amplitudes that come from graphs that satisfy the
    condition. 
  \item Hypothesis \ref{hyp}.(2) is sometimes assumed without the
    requisite advertisement. It is always in force when the external
    momenta are in the Euclidean region. Moreover, for the purpose of
    studying Feynman integrals as a family (for example, via
    GKZ-systems), momenta can be viewed as generic (subject to the
    external momentum sum being zero), and then Hypothesis
    \ref{hyp}.(2) holds as well.
  \item If Hypothesis \ref{hyp}.(3) is violated, the problem is
        trivialized to $\calG_m=\calU$ in which case it is known that
        the semigroup spanned by its support vectors is normal, \cite{HT}.
  \end{asparaenum}
\end{rmk}

\medskip

Treating the nonzero masses and momenta as indeterminates, one arrives
at a differentiable family of integrals. One method to study Feynman
integrals is by computing differential equations that govern this
family, and then solving them with a
power series Ansatz. After that, one may consider the specialization
of certain variables to special values, or one can investigate
geometric behavior (such as momodromy) of the family.

Let
\[
A_m:=\begin{pmatrix}
1&1&\cdots&1&1\\ \bolda_1&\bolda_2&\cdots&\bolda_{n-1}&\bolda_n
\end{pmatrix}
\]
be the matrix the columns of which are given by the lifted exponents
$\bolda_i$ of the monomials $\bsx^{\bolda_i}$ appearing in $\calG_m$;
here and elsewhere we call $(1,\bolda)\in\ZZ\times \ZZ^E$ the
\emph{lift} of $\bolda\in\ZZ^E$.  More generally, let $A$ be any
integer $(1+|E|)\times n$ matrix. We shall refer to
the group of integer linear combinations of the columns of $A$,
\[
\ZZ A:= \{\sum m_i \bolda_i\,\, |\,\, m_i\in\ZZ\}
\]
as the \emph{lattice of $A$}. In conjunction with
any choice of a complex parameter vector $\beta\in\CC\times
\CC^{E}$, such matrix $A$ induces a \emph{GKZ-system} (or also
called \emph{$A$-hypergeometric system}) $H_{A}(\beta)$ of linear
partial differential equations in $n$ new variables $y_1,\ldots,y_n$,
as we explain in the next section.  It is known for $A=A_m$ that a
suitable choice of the parameter $\beta$ causes the
$A_m$-hypergeometric system $H_{A_m}(\beta)$ to have among its
solutions the family of Feynman integrals to the graph $G$; see
\cite{Klausen-JHEP20,HT} for a down-to-earth discussion on this.

In the construction of the hypergeometric system $H_A(\beta)$ enters a
certain toric ideal $I_A$ in the polynomial ring $R_{A}=\CC[\bsdel]$
in the partial differentiation operators $\del_1:=\frac{\del}{\del
  y_1},\ldots,\del_n:=\frac{\del}{\del y_n}$; it is induced by the
monomial map from $\CC^*\times(\CC^*)^{E}$ to $\CC^n$ encoded in
$A$. Let $I_{A}$ be the ideal of $R_A$
describing the closure of the image of $\CC^*\times(\CC^*)^{E}$ in
$\CC^n$.
If the quotient
\[
S_A:=\CC[\NN A]\simeq R_A/I_A
\]
enjoys a certain algebraic property known as
\emph{Cohen--Macaulay}, then various desirable simplifications
regarding the solutions of $H_{A}(\beta)$ occur. As is discussed in 
\cite{Klausen-JHEP20}, of practical value in the theory of Feynman
integrals are: suitable initial ideals of $H_{A}(\beta)$ become
computable in elementary fashion without need to look at Gr\"obner
bases, and classical combinatorial recipes for manufacturing solutions
become much simpler, see \cite{SST} for background on hypergeometric
differential equations.

The Cohen--Macaulayness of $S_A$ is implied by, but by no means
equivalent to, the condition that the semigroup $\NN A\subseteq
\RR\times \RR^E$ be \emph{saturated}, which means that the
intersection of the non-negative rational cone $\RR_{\geq 0}A$ spanned
by the columns of $A$ over the origin with the lattice $\ZZ A$
contains no other lattice points than those in $\NN A$; see
\cite{SST,MMW} for more details on Cohen--Macaulayness in this
context.  Saturatedness is an arithmetic condition that involves study
of the interior points of the dilations of the polytope spanned by the
columns of $A$.

For notation, let the \emph{support} $\Supp(f)$ of a Laurent
polynomial $f=\sum c_\bolda \bsx^\bolda$ be the exponent vectors
\[
\Supp(f):=\{\bolda\mid c_\bolda\neq 0\}
\]
of the monomials appearing with nonzero coefficient in $f$.  Denoting
the convex hull of a set $S\subseteq \RR^E$ by $\overline S$, the
\emph{support polytope} of $f$ is $\overline{\Supp(f)}$. Let $P_m$ be
the support polytope of $\calG_m$.  Helmer and Tellander \cite{HT}
showed in the following two extreme cases that the semigroup of $A_m$
is saturated:
\begin{enumerate}
  \item in  the \emph{massive case} where each particle mass is
    positive,
  \item in the \emph{massless case} where each particle mass is
    zero (with the additional assumption that every vertex is external).
\end{enumerate}
In both cases, this
means that $S_{A_m}$ is Cohen--Macaulay.
The tools they use include
edge-unimodularity, flag matroid polytopes,
Cayley and Minkowski sums, which they used to study IDP properties of
polytopes.

In this note, we start with discussing the support vectors of
$\calG_m$ from the point of view of matroid theory. Of course, the
support vectors of $\calU$, interpreted as indicator functions,
describe the co-graphical matroid of $G$. We show here that the
support vectors of $\calF_0$ and those of the square-free terms in
$\calU\cdot(\sum_{e\in G}m_e^2x_e)$ both describe matroids as well. We
also show that, quite surprisingly, their union forms a matroid as well. So,
for all Feynman graphs the support vectors of the square-free terms of
$\calF$ form a matroid.

We use these matroidal results, and some ideas of \cite{HT} to show
that the semigroup generated by $A_m$ is saturated for (s1I) graphs
$G$ in the two cases
\begin{enumerate}
  \item if every 2-forest of $G$ induces a nonzero term in $\calG_m$
    (Theorem \ref{thm-main});
  \item if $m_e=0$ for all $e$ (Theorem \ref{thm-massless});
\end{enumerate}
which  generalize the two corresponding cases in
\cite{HT}.
In consequence, in these cases $A_m$ defines a hypergeometric system that
enjoys the Cohen--Macaulay property.

\medskip

In the next section we set up the necessary notation, and carefully
describe the needed details about hypergeometric systems, as well as
graphs, polytopes and matroids. In Section 3, we discuss the
advertised matroids, and in Section 4 we state and prove the semigroup
results. Under Condition (1) above, this follows from an inspection of
the way that the cone over $A_m$ behaves under specialization of a
mass to zero. In the massless case we follow the route of \cite{HT} in
the corresponding context. We also proide some partial results towards
the general case. In the last section we discuss some examples of the failure
of Hypothesis \ref{hyp}. For the convenience of the reader, we provide
a list of symbols at the end.

\section*{Acknowledgements}
I am much indebted to Ren\'e-Pascal Klausen for an astute observation
leading to Hypothesis \ref{hyp}(2), for enlightening discussions on
Feynman amplitudes and QFTs, and for criticism on earlier versions of this
article. My sincere thanks go to Martin Helmer and Felix Tellander
writing their article and for sharing their insights.  I am also
grateful to Diane MacLagan, Christian Haase and Karen Yeats for
helpful explanations on polytopal yoga.

\section{Notation and basic concepts}

If $e\in E$
then we denote the unit vector of $\RR^E$ pointing in $e$-direction by
$\bolde_e$, and if $S\subseteq E$ is a collection of edges then we
write $\boldv_S$ for the \emph{indicator vector} of $S$ defined by 
\[
\boldv_S=\sum_{e\in S}\bolde_e.
\]

\subsection{Hypergeometric systems}

We give here a minimal introduction to $A$-hypergeometric systems
invented by Gel'fand, Graev, Kapranov and Zelevinsky in the mid-1980s.
For details and literature on them and on parametric integrals that
occur as their solutions we refer to the book \cite[Sec.~5.4]{SST},
and to the survey \cite{RSSW}.

Take an integer matrix $A\in\ZZ^{(1+d)\times n}$, and a set of
variables $\bsy=y_1,\ldots,y_n$. Denote the partial derivative
operators $\del/\del y_j$ by $\del_j$ and consider the Weyl algebra
$D_A$ in variables $y_1,\ldots,y_n$ given as the non-commutative ring
$\CC[\bsdel]\langle \bsy\rangle$. The elements of $D_A$ can be
interpreted as linear differential operators in $\bsy$ with polynomial
coefficients.

The matrix $A$ induces a monomial action
\begin{eqnarray*}
(\CC^*)^{1+d}\times \CC^n&\to&\CC^n,\\
  (\frakt,\fraky)&\mapsto&(\frakt^{\bolda_1}\fraky_1,\ldots,\frakt^{\bolda_n}\fraky_n)
\end{eqnarray*}
of the $(1+d)$-torus on the affine space with coordinates
$\del_1,\ldots,\del_n$.  The usual closure of the orbit of the point
$(1,\ldots,1)\in\CC^n$ is also Zariski closed, and defined by the
toric ideal $I_A$ generated by the binomials
$\Box_{\boldu,\boldv}:=\bsdel^\boldu-\bsdel^\boldv$, running over all
$\boldu,\boldv\in\NN^n$ with $A\cdot\boldu=A\cdot\boldv$. One may view
$I_A$ as a subset of $D_A$ via the embedding of rings
$\CC[\bsdel]\into D_A$.

The matrix $A$ also induces $(1+d)$ \emph{Euler operators}
\[
\qquad\qquad E_i:=\sum_{j=1}^n a_{i,j}y_j\del_j\in D_A\qquad\text{ for
  $0\le i\le d$}.
\]
Given a choice of $\beta\in\CC^{1+d}$, the \emph{hypergeometric ideal} to $A$
and $\beta$ is
\[
H_A(\beta):=D_A\cdot(I_A,\{E_i-\beta_i\}_{i=0}^{d}).
\]
Any left ideal $H=\sum D_A Q_i$ of $D_A$ generated by the operators
$\{Q_i\}_i\subseteq D_A$ can be interpreted as a system of linear
partial differential equations on a solution function $\phi(\bsy)$, by
asking that $Q\bullet (\phi(\bsy))=0$ for all $Q\in H$ (or,
equivalently, that $Q_i\bullet(\phi(\bsy))=0$ for all $i$). As is
explained in \cite{HT}, if one reads the coefficients of $\calG_m$ as
parameters then the Feynman integrals
corresponding to $A_m$ appear as solutions of
$H_{A_m}(\beta)$ for the right choice of $\beta$. For the study of
Feynman integrals, the entire family is useful; for some purposes 
even  $\beta$ is viewed as a variable.

\begin{rmk}
  A frequent hypothesis in the theory of $A$-hypergeometric systems is
  that the group $\ZZ A$ generated by the columns of $A$ agrees with
  the ambient lattice $\ZZ^{1+d}$ inside $\RR^{1+d}$. The hypothesis
  is not crucial to the majority of known results, but it usually allows a
  much simpler formulation. However, the question whether a semigroup
  ring is normal is only decided by the saturatedness of the semigroup
  in its own lattice, the group it generates.
\end{rmk}

\subsection{Polytopes}

A polytope $P$ in $\RR^{1+d}$ is a \emph{lattice polytope} if its
vertices belong to the lattice $\ZZ\times\ZZ^d$ inside $\RR^{1+d}$. 

Given two polytopes $P, P'$ in $\RR^E$,
their \emph{Minkowski sum}
$P+P'$ is the set of points $\{w=v+v'\in \RR^E\mid v\in P, v'\in
P'\}$. The edges of a Minkowski sum are parallel to egdes
of the input polytopes.
The vertices of a Minkowski sum
are always sums of vertices of the input polytopes
(although some such sums might be interior points of the sum
polytope). In contrast, the set of the lattice points in a Minkowski
sum is often not equal to the sum of the sets of lattice points
in the two input polytopes.

Let us set 
\[
E_m:=\{e\in E\mid m(e)\neq 0\}\qquad \text{and} \qquad E_0:=E\minus E_m.
\]
Writing $m_e$ for $m(e)$ to ease notation, set
\[
\Sigma_m:=\sum_{e\in E_m}m_e^2x_e\quad \text{ and }\quad
\Delta_m:=\overline{\Supp(\Sigma_m)};
\]
the latter is the simplex in $\RR^E$ spanned by the unit vectors
$\{\bolde_e\}_{e\in E_m}$.

We also set
\[
\tilde\Sigma_m:=1+\Sigma_m,\qquad\text{ and }\qquad
\tilde\Delta_m:=\overline{\Supp(\tilde\Sigma_m)}.
\]

If we already have a specific mass function $m$ in mind, we write
\begin{eqnarray}\label{eqn-withE1}
  \Sigma_E:=\Sigma_m+\sum_{e\in E_0}x_e,&\qquad&\Delta_E:=\Supp(\Sigma_E),\\
  \tilde\Sigma_E:=1+\Sigma_E,&&\tilde\Delta_E:=\Supp(\tilde\Sigma_E).\label{eqn-withE2}
\end{eqnarray}

According to Hypothesis \ref{hyp}, the support polytope of $\calG_m$
is the same as the polytope spanned by the union
$\Supp(\calU)\cup\Supp(\calU\cdot\Sigma_m)\cup
\{\Supp(|p(F)|^2\cdot \bsx^F)|F\in\calT^2_G\})$ since in the sum
$\calG_m=\calU\cdot\tilde\Sigma_m+\calF_0$ no terms are lost due to
coefficient cancellations,

\subsection{Graphs and their matroids}\label{subsec-matroids}

We generally use the graph and matroid language as it prevails in
mathematics. So, for us a \emph{loop} is an edge that is incident to
only one vertex; a
\emph{circuit} is a set of edges whose union in a realization of the
graph is homeomorphic to a polygon (in physics this is sometimes called
a \emph{loop}).



In each term of $\calU$ and of $\calF_0$, each variable appears (by
definition) with degree at most one. On the other hand,
$\calU\cdot\Sigma_m$ can have some terms with some variable of degree
two (and the other variabless of degree one or zero). Such square
terms can occur only for massive variables (and if a variable is in
fact massive then it will occur in some term with degree two since the
corresponding edge cannot not belong to every 1-forest in the (s1I)
graph $G$).

A \emph{matroid} $\matroidM$ is determined by a distinguished
collection $\calB_\matroidM\subseteq 2^E$ of \emph{bases}, all of
equal cardinality, taken from a fixed ground set $E$.
From this angle, the defining property of a  matroid is a version of the Exchange Axiom of Steiner from
linear algebra: if $B,B'$ are two bases of a matroid, and $e\in B$,
then there is $e'\in B'$ such that $(B\minus \{e\})\cup \{e'\}$ is
again a basis. In
fact, there is an equivalent ``strong'' version where in the same
notation the set 
$(B'\minus\{e'\})\cup\{e\}$ can also be arranged to be a basis.
The
notion of a matroid generalizes the idea of linear independence of
sets of vectors, and much of the nomenclature is borrowed from linear
algebra. We refer to \cite{Oxley} for background and all facts that we
use about matroids.

For example, matroids have a \emph{rank function}
\[
\rank_\matroidM\colon 2^E\to\NN,
\]
and the bases are precisely the minimal sets (with respect to
inclusion) of maximum possible rank in $\matroidM$. The rank of a
matroid is (by definition) the size of any of its bases (which is
indeed a well-defined integer).  A
\emph{loop} of a matroid $\matroidM$ is an element $e$ for which
$\rank_\matroidM(\{e\})=0$.  To each basis $B$ one has an indicator
vector $\boldv_B$ in $\{0,1\}^E$ with $\boldv_B(e)=1$ if and only if
$e\in B$; so the entry sum of any $\boldv_B$ is the rank
of $\matroidM$. The convex hull of the lattice vectors $\{\boldv_B\mid
B\in\calB_\matroidM\}$ is the \emph{matroid polytope} of $\matroidM$.

Every $\boldv_B$ is a vertex of the matroid polytope, since it is even
a vertex of the polytope spanned by all integer vectors that have only
0/1 entries and entry sum $\rank(\matroidM)$. Indeed, among such
integer vectors, $\boldv_B$ realizes the unique maximum of the linear
function that takes dot product against $\boldv_B$. 

A matroid is
\emph{Boolean} if $E$ itself is a basis (and then the only one). More
generally, the Strong Exchenge Axiom implies that
 the
edges of the matroid polytope are  precisely those that link
(indicator vectors of) bases
that agree in all but two positions. In particular,
egdes of the matroid polytope are parallel to the vectors
$\bolde_e-\bolde_{e'}$, \cite{GGMS}.

A \emph{circuit} of a matroid is a set that is not contained in any
basis, and minimal (with respect to inclusion) in this regard. Loops
are circuits.  
An \emph{independent} set is one that contains no circuit;
independent sets are exactly those subsets of
$E$ on which the
rank function agrees with the cardinality function, and they can also
described as the sets that are subsets of bases. Bases are maximal
independent sets, and proper subsets of circuits are independent.

\medskip

If $G$ is a graph, the collection $\calT^1_G$ of 1-forests of
$G$ forms the set of bases for a matroid $\matroidM^1_G$ on the
underlying set $E$ of edges. Circuits of the graph are then circuits
of $\matroidM^1_G$, and (graph-theoretic) loops correspond to
(matroid-theoretic) loops. Matroids that arise
this way are called \emph{graphic}.

For a set of edges $S$ from $G$ (which we read as a subgraph of $G$ on
the same vertex set $V_G$) we call their \emph{span} the collection of all
edges of $G$ that connect vertices of $G$ that belong to the same
connected component in the subgraph $S$. In other words, the vertex
partitions of $V_G$ by sets of
connected components of $S$ and $\Span(S)$ are the same, and $\Span(S)$
is the largest subgraph of $G$ in this regard. Then
$\rank(S)=\rank(\Span(S))$ is the difference of the number of
components of $S$ (as graph on the vertex set of $G$) and $|V_G|$.
The rank function can also be interpreted as 
the size of the largest circuit-free subset, and span in a general
matroid is the largest superset with the same rank as the given set.

The set of complements $\{E\minus T\mid T\in\calT^1_G\}$ forms the set
of bases for another
matroid $\matroidM^{1,\perp}_G$ on $E$ that turns out to be dual
to $\matroidM_G$ in a suitable sense.
For this \emph{cographic} matroid $\matroidM_G^\perp$, a loop is an edge
that is part of every 1-forest of $G$. Its removal thus disconnects the
graph and such edge cannot occur in a (s1I) Feynman diagram. So, for an
(s1I) graph, neither the graphic nor the cographic matroid has
loops.

Similarly, the set of 2-forests $\calT^2_G$, as well as the set of
their complements, form matroids that we denote $\matroidM^2_G$ and
$\matroidM^{2,\perp}_G$ respectively. 

Any matroid can be written as a matroid sum of simple matroids; a
matroid is \emph{simple} if
 it is impossible to write the set of bases $\calB_\matroidM$ as the
 set of unions of the bases of two submatroids on disjoint subsets
 $E_1,E_2$ of $E$. A graph is (s1I) if and only if its graphic 
 and cographic matroid are
 simple.

Let $\bsx=\{x_e\mid e\in E\}$ be a set of indeterminates that are in
correspondence with the elements of the ground set of
$\matroidM$. There is
an induced \emph{matroid basis polynomial}
\[
\Phi_\matroidM=\sum_{B\in\calB_\matroidM}\bsx^{\boldv_B}\in\CC[\bsx]
\]
with very interesting combinatorial properties.\footnote{A more general
   class of polynomials arises from \emph{realizations} of
  matroids, see for example \cite{BEK,Patterson,DSW,DPSW}.}
The polynomial $\calU$ is the matroid basis polynomial
$\Phi_{\matroidM^{1,\perp}_G}$ of $\matroidM_G^{1,\perp}$, and the induced
polytope
\[
P_G^{1,\perp}:=\overline{\Supp(\calU)}
\]
is the matroid polytope to $\matroidM_G^{1,\perp}$.  On the other hand,
$\Delta_E$ is the matroid polytope to the cographic matroid on $E$
corresponding to a polygon with $|E|$ edges (or to the graphic matroid
to the graph on $|E|$ edges with only two vertices and no loops; these
are called \emph{banana} or \emph{sunset} graphs).

\bigskip

If $\matroidM$ is any matroid on the set $E$, then the semigroup
\[\{\boldv_b\mid B\in \calB_\matroidM\}\]
is saturated in its own lattice, by \cite[Thms.~1, 2]{White}. 

For any pointed (\emph{i.e.}, no invertibles except for the neutral
element) sub-semigroup $S$ of a free Abelian group of finite
rank, the semigroup ring $\CC[S]$ is normal if and only if $S$ is
saturated in the group generated by $S$ according to \cite{Hochster},
and this happens if and only if the semigroup of its lifts is
saturated in the ambient lattice. All such semigroup rings are toric,
and therefore their normality implies Cohen--Macaulayness.

  By a \emph{unimodular matrix} we mean here a matrix with integer
entries whose maximal minors are all in the set $\{-1,0,1\}$. 

In \cite{HT} it is shown that the
semigroup generated by the lifts of the support vectors of $\calG_m$ is
normal provided that either a) all masses are nonzero, or b) all
masses are zero and every vertex is an external vertex.

\section{Matroids in Feynman Theory}\label{sec-matroids}

Recall that we assume that $G$ satisfies the conditions in Hypothesis
\ref{hyp}, and that $\calT^1_G$ and $\calT^2_G$ denote the collections
of spanning trees and 2-forests of $G$ respectively.

By Hypothesis \ref{hyp}, the monomials appearing in $\calG_m$ are
exactly those appearing $\calU$, plus those appearing in either
$\calU\cdot \Sigma_m$ or $\calF_0$.  The squarefree ones in these last
two polynomials are indexed, respectively,
by a massive edge in a spanning tree for $G$, or a 2-forest with
non-vanishing moment coefficient.  In this section we investigate
the matroidal properties of these two sets. They form the tools for
the main results in the next section.

In order to simplify
the discussion we introduce some language.
\begin{ntn}
  If $G',G''$ are subgraphs of $G$ then if $e\in E_G$ is an edge we
  say it \emph{links $G'$ to $G''$} if it involves one vertex from
  $G'$ and one vertex from $G''$. We further say that $e$ \emph{is
    supported on $G'$} if both vertices of $e$ are vertices of
  $G'$. This does not require that $e$ be an edge of $G'$.
\end{ntn}

\subsection{Momentous 2-forests}

\begin{dfn}
  A 2-forest $F\in\calT^2_G$ is \emph{momentum-free} if the momentum
  coefficient $|p(F)|^2$ of $\bsx^{E\minus F}$ in $\calF_0$ is zero. We denote the set of
  momentum-free 2-forests of $G$ by $\calT^2_{G,0}$. 

  We call the elements of the complementary set
  \[
  \calT^2_{G,\neq}:=\calT^2_G\minus \calT^2_{G,0}
  \]
  the \emph{momentous 2-forests}.
  
Note that, by Hypothesis \ref{hyp},  a 2-forest $F=F_1\sqcup F_2$ with connected
components $F_1,F_2$ is in $\calT^2_{G,0}$ precisely when either
$\VExt\subseteq F_1$ or $\VExt\subseteq F_2$.

\end{dfn}

As a very special example of a momentum-free 2-forest, let $v$ be an
interior vertex and let $F$ be a spanning tree for the graph obtained
by deleting $v$ and all incident edges from $G$. Then $F\cup\{v\}$ is
a 2-forest for $G$ that lies in $\calT^2_{G,0}$. More extremely, if
$G$ were permitted to have only one external vertex, no momentous
2-forest would exist at all, and $\calF_0$ would be zero altogether.

\begin{lem}\label{lem-moment}
  The set $\calT^2_{G,\neq}$
  is  the set
  of bases of a matroid on the edge set $E$ of $G$.
\end{lem}
\begin{proof}
  If $|\VExt|=1$, there are no momentous 2-forest, so there is nothing
  to show. So we assume that at least two external vertices exist.
  
  If $\calT^2_{G,\neq}$ is non-empty, we need to show that the set of
  momentous 2-forests satisfies the matroid basis exchange axiom.  So,
  choose $F\in\calT^2_{G,\neq}$, and suppose $F'$ is an arbitrary
  second 2-forest. Choose $e\in F$; then $F\minus \{e\}$ is a 3-forest
  $F_1\sqcup F_2\sqcup F_3$ of $G$, where the $F_i$ are the connected
  components of $F\minus \{e\}$.

  Since the full collection $\calT^2_G$ of 2-forests forms the set of
  bases of a matroid, some edges of $F'$, when added to
  $F\minus\{e\}$, produce again a 2-forest. These are precisely those
  edges of $F'$ that link $F_i$ to  $F_j$,
  for $i\neq j$ in $\{1,2,3\}$.

  Since $F$ is in $\calT^2_{G,\neq}$, the external vertices do not lie
  entirely inside one of the components of $F$, and even less do they
  lie entirely inside a connected component $F_i$ of $F\minus
  \{e\}$. Thus, after possibly relabeling, both $F_1$ and $F_2$, and
  possibly also $F_3$, will contain an external vertex. If $F_3$ does
  in fact contain an external vertex, then adding any edge $f\in F'$
  to $F_1\sqcup F_2\sqcup F_3$ will leave the external vertices split
  between at least two different connected components. Combined with
  the previous paragraph and Hypothesis \ref{hyp}.(3) we can dispose
  of the case when $F_3$ also contains an external vertex.

  Now suppose $F_3$ does not contain an external vertex, so $\VExt$ is
  in the disjoint union $ F_1\sqcup F_2$. If $F'$ contains an edge $f$
  that links $F_3$ either to $F_1$ or to
  $F_2$, we are done, since then $(F\minus \{e\})\cup \{f\}$ is a
  2-forest in $\calT^2_{G,\neq}$. So consider the possibility that
  $F'$ has no such edge; then no edge of $F'$
  links  $F_3$ to
  $F_1\cup F_2$.  This disconnection shows that the 2-forest
  $F'$ has one connected component that uses the vertices of $F_3$,
  and one component that uses the vertices of $F_1\cup F_2$. But then
  $F'$ has $\VExt$ inside one of its components and thus can't be in
  $\calT^2_{G\neq}$. The lemma follows.
\end{proof}

\begin{dfn}
  We denote the matroid of the previous lemma by
  $\matroidM^2_{G,\neq}$.
\end{dfn}

Recall that a matroid $\matroidM'$ is a quotient of the matroid
$\matroidM$ if (they are matroids on the same ground set and) any
circuit in $\matroidM$ is  a union
of circuits in $\matroidM'$. 

\begin{lem}
  $\matroidM^2_{G,\neq}$ is a quotient of $\matroidM^1_G$.
\end{lem}
\begin{proof}

The graphic matroid $\matroidM^1_G$ of $G$ has as circuits the circuits of $G$.
Suppose $C$ is one such circuit; it cannot be independent in
$\matroidM^2_{G,\neq}$ since it cannot be contained in any
2-forest. We will show that it is the union of circuits in
$\matroidM^2_{G,\neq}$.

  If $\matroidM^2_{G,\neq}$ is the trivial matroid, each singleton is
  a circuit, and the lamme follows. So, we can assume that
  $\matroidM^2_{G,\neq}$ is not trivial.

For the moment assume that $C$ contains at least one, but not every, external
vertex. Let $e$ be any edge of $C$.  As $C\minus \{e\}$ is
independent in $\matroidM^1_G$, we can embed it into a spanning tree $T$
for $G$.  Then let $v$ be an external vertex not in $C$. Since the set
$C\minus\{e\}$ is connected and $T$ is a tree, there is a unique
shortest path in $T$ that connects $v$ with $C\minus \{e\}$. Remove
one of the edges $f$ in this shortest path to obtain from $T$ a
2-forest $F$ in which $v$ and $C\minus \{e\}$ lie in different
components. Note that by construction $f$ is not in $C$. It
follows that $F$ is a basis in in $\matroidM^2_{G,\neq}$ and so
$C\minus\{e\}\subseteq F$ is independent in
$\matroidM^2_{G,\neq}$. Since this is so for any $e\in C$, $C$ is a
circuit in $\matroidM^2_{G,\neq}$.

Now suppose $C$ contains no external vertex. Again, remove an
arbitrary edge
$e\in C$ and embed the resulting $C\minus \{e\}$ into a spanning tree
$T$ for $G$. Choose any two external vertices $v,v'$. Within $T$ there
is a unique minimal path from $v$ to $v'$. Since neither vertex is in
$C$, there is at least one edge $f$ in this minimal path
that does not belong to $C$. Remove $f$ from $T$ to arrive
at a 2-forest containing $C\minus\{e\}$. It is momentous by Hypothesis
\ref{hyp}.(3) since the external vertices are not all in one
component. It follows that
removing any edge from $C$ makes it independent in
$\matroidM^2_{G,\neq}$ and thus $C$ is a circuit in
$\matroidM^2_{G,\neq}$.

Finally, suppose $C$ contains all $\ell\geq 2$ external vertices.
Denote the vertices of $C$ by $v_1,\ldots,v_c$, written in such a way
that $(v_j,v_{j+1})$ are the edges of the circuit (with the
understanding that $v_{c+1}=v_1$). Let $1\le i_1<\ldots<i_\ell\le c$
be the labels that correspond to the $\ell=|\VExt|$ external
vertices. Let $C_k$ be the result of removing from $C$ the edges
$(v_{i_k},v_{i_k+1}),\ldots,(v_{i_{k+1}-1},v_{i_{k+1}})$ that lie in
the chosen orientation of $C$ between the external vertices $v_{i_k}$
and $v_{i_{k+1}}$. (Again, we agree that
$v_{i_{\ell+1}}=v_{i_1}$). Then in $\matroidM^1_G$, these sets $C_k$
are independent, but in $\matroidM^2_{G,\neq}$ they are still
dependent since they contain all external vertices. We claim that
$C_k$ is in fact a circuit in $\matroidM^2_{G,\neq}$. Indeed, for any
edge $e\in C_k$, the graph $C_k\minus\{e\}$ has two connected components and
$\VExt$ is not contained in either one: one component contains
$v_{i_k}$ and the other contains $v_{i_{k+1}}$. Thus, $C_k\minus\{e\}$
can be completed to a 2-forest such that neither of its components
contains $\VExt$ and is hence independent in
$\matroidM^2_{G,\neq}$. To finish the proof, observe that $C$ is
covered by the various $C_k$.

\end{proof}

\subsection{Massive truncations}

\begin{dfn}
  A 2-forest $F$ that can can be written as $T\minus\{e\}$ for a
    spanning tree $T$ and a massive edge $e$ is called a \emph{massive
      truncation (of $T$ by $e$)}. We denote by $\calT^2_{G,\mt}$ the
    collection of massive truncations.  
\end{dfn}
The massively truncated 2-forests are those that label nonzero
squarefree terms in $\calF^W_0$.
\begin{lem}\label{lem-mass-trunc}
  The set of massively truncated 2-forests forms the
  bases of a matroid on the edge set $E$ of $G$.
\end{lem}
\begin{proof}
We need to show that the set of massively truncated 2-forests, if
non-empty, satisfies the Exchange Axiom.
  
  Let $F,F'$ be massively truncated 2-forests and choose massive edges $e,e'$
  such that $T=F\cup\{e\}$ and $T'=F'\cup \{e'\}$ are spanning trees. Let
  $f\in F$ and consider the 3-forest $F\minus\{f\}$ with connected
  components $F_1,F_{2a},F_{2b}$ where $F_1$ is one component of $F$ and
  $F_{2a}\cup F_{2b}\cup\{f\}$ is the other. We need
  to show that for a suitable $g\in F'$, the set $(F\minus
  \{f\})\cup\{g\}$ is a massively truncated 2-forest.

  Since the 2-forests of $G$ form a matroid $\matroidM^2_G$, certain
  edges $g$ of $F'$ must combine with $(F\minus \{f\})$ to a
  2-forest. Moreover, the edge $e=T\minus F$ from above 
  links $F_1$ to either $F_{2a}$ or
  $F_{2b}$; without loss of generality we can and do assume that $e$
  links in fact $F_1$ to $F_{2a}$. 

  If some edge $g$ of $F'$ links a vertex of $F_{2a}$ to a vertex of
  $F_{2b}$, then $(F\minus\{f\})\cup\{g\}$ is a 2-forest on the same
  connected components as $F$ and thus can be completed by the massive
  edge $e$ to a spanning tree. Similarly, if any edge $g$ of $F'$
  links $F_{1}$ to $F_{2b}$, then
  $(F\minus\{f\})\cup\{g\}$ is a 2-forest in which $F_{2a}$ is a
  connected component and again the 2-forest $(F\minus\{f\})\cup\{g\}$
  can be completed by the massive edge $e$ to a spanning tree.  So,
  assume from now on that $F'$ has no edges from $F_{2a}$ to $F_{2b}$,
  and no edges from $F_1$ to $F_{2b}$.

  In that case, the vertices of $F_{2b}$ must be exactly the vertices
  in one of the two components of the 2-forest $F'$ and therefore the
  other component of $F'$ uses exactly the vertices of $F_1\cup
  F_{2a}$.  In particular there is guaranteed to be an edge $g$ in
  $F'$ from a vertex of $F_1$ to a vertex of $F_{2a}$. Note that
  $(F\minus \{f\})\cup\{g\}$ is then a 2-forest. Now recall that
  $F'=T'\minus\{e'\}$ is a massive truncation. Clearly, $e'$
  must connect the two components of $F'$ and so links 
  $F_{2b}$ to either $F_1$ or $F_{2a}$. In that case,
  $(F\minus \{f\})\cup\{g\}$ is a massive truncation by $e'$.
\end{proof}

\begin{dfn}
  We denote the matroid of massively truncated 2-forests of $G$
  from Lemma \ref{lem-mass-trunc} by $\matroidM^2_{G,\mt}$.
\end{dfn}

We show next that the matroid of massively truncated 2-forests is also
quotient of $\calT^1_G$, but we use a different strategy than for the
momentous 2-forests.

\begin{dfn}
Suppose $\matroidM$ is a matroid on the set $E$ and $ E'\subseteq
E$. Define $\calB_{/E'}$ to be the set of subsets $B$ of $E$ that have
the property that there is some $e'\in E'\minus B$ such that $B\cup
\{e'\}$ is a basis in $\matroidM$.

By Lemma \ref{lem-M/E'} below, the sets in $\calB_{/E'}$ are the bases
of a matroid that we denote $\matroidM_{/E'}$. Note that when $E'$ has
rank zero (so $E'$ contains only loops) then $\matroidM_{/E'}$ is the
trivial matroid.
\end{dfn}

\begin{lem}\label{lem-M/E'}
  The set $\calB_{/E'}$ is the set of bases of a matroid.
\end{lem}
\begin{proof}
  Let $B,B'\in\calB_{/E'}$ and choose $f\in B$. Let $e\in E'\minus B$ and
  $e'\in E'\minus B'$ be
  such that $B\cup\{e\}, B'\cup\{e'\}$ are bases for
  $\matroidM$. Then for some element $g$ of
  $B'\cup\{e'\}$ the Exchange Axiom in $\matroidM$ guarantees that
  $((B\cup\{e\})\minus\{f\})\cup\{g\}$ is a basis for
    $\matroidM$. Since $f\neq e\neq g$, $(B\minus\{f\})\cup\{g\}$ is the new
      basis for $\matroidM_{/E'}$ that we want.
\end{proof}

\begin{rmk}
  Note that the independent sets of $\matroidM_{/E'}$ are those
  contained in a basis of $\matroidM_{/E'}$ and therefore are the
  susets of $E$ that can be augmented to an independent set in
  $\matroidM$ by an element of $E'$. In particular, if $e'\in E'$ is a
  loop, then $\matroidM_{/E'}=\matroidM_{/(E'\minus\{e'\})}$.
\end{rmk}

\begin{lem}
  The matroid $\matroidM_{/E'}$ is a quotient of the matroid
  $\matroidM$.
\end{lem}
\begin{proof}
  We shall asssume that $E'$ has positive rank.  
  Let $C$ be a circuit of $\matroidM$; in particular, $|C|=\rank(C)+1$.
  
  Suppose first, that the span of $C$ does not contain all of $E'$,
  and choose $e'\in E'\minus C$. Select $c\in C$. Then $|(C\minus
  \{c\})\cup\{e'\}|= |C|=\rank(C)+1=\rank(C\cup\{e'\}) =\rank((C\minus
  c)\cup \{e'\})$. It follows that $(C\minus \{c\})\cup \{e'\}$ is
  independent in $\matroidM$ and hence $C\minus \{c\}$ is independent
  in $\matroidM_{/E'}$. Thus, in this case $C$ is a circuit in
  $\matroidM_{/E'}$.
  
  Now suppose $E'\subseteq \Span(C)$. Let $C'$ be a minimal subset of
  $C$ such that $\Span(C')$ contains $E'$. Note that $C'$ is dependent
  in $\matroidM_{/E'}$. Take $c'\in C'$; then $\Span(C'\minus\{c'\})$
  does not contain $E'$ and so there is $e\in E'$ such that
  $\rank(C'\minus\{c'\})<\rank((C'\minus\{c'\})\cup \{e'\})$. Since
  $C'\minus\{c'\}$ is indepedent in $\matroidM$, so that rank agrees
  with cardinality, the same is true for $(C'\minus\{c'\})\cup \{e'\}$ and
  so this set must be independent in $\matroidM$. It follows that the
  $\matroidM_{/E'}$-dependent set $C'$ is covered by the
  $\matroidM_{/E'}$-independent sets $C'\minus\{c'\}$ and thus a
  $\matroidM_{/E'}$-circuit.

  Next, consider $C\minus \{c'\}$ for $c'\in C'$. Since $C$ is an
  $\matroidM$-circuit, $\Span(C\minus\{c'\})=\Span(C)\supseteq E'$ and
  it follows that $C\minus\{c'\}$ is $\matroidM_{/E'}$-dependent. For
  any $c\neq c'$ in $C$, we have $\rank((C\minus \{c,c'\})\cup
  E')>\rank(C\minus \{c,c'\})=|C|-2$. (Indeed, a flat that contains
  $C'\minus \{c'\}$ and $E'$ must contain $C'$, and so any flat that
  contains $(C\minus \{c,c'\})\cup E'$ must contain $C\minus\{c\}$, of
  rank $|C|-1$). It follows, that for a suitable element $e'\in
  E'\minus(C\minus \{c,c'\})$, $(C\minus \{c,c'\})\cup \{e'\}$ is
  $\matroidM$-independent. It follows that $C\minus\{c,c]\}$ is
      $\matroidM_{/E'}$-independent, and so $C\minus \{c'\}$ is a
      $\matroidM_{/E'}$-circuit.

  The union of all circuits discussed covers $C$.
\end{proof}

\begin{cor}
  The matroid $\matroidM^2_{G,\mt}$ is a quotient of $\matroidM^1_G$.
\end{cor}
\begin{proof}
  In the previous lemma, take $\matroidM=\matroidM^1_G$ and $E'$ to be
  the massive edges. Then the definition of $\matroidM^2_{G,\mt}$
  matches that of $(\matroidM^2_G)_{/E'}$.
\end{proof}

\begin{rmk}
  The term ``quotient matroid'' comes historically from the following
  quotient construction of vector spaces.  Suppose $R=\{r_e\}_{e\in
    E}\subseteq \RR^n$ is a collection of vectors. It is often
  referred to a \emph{representation of a matroid}. Indeed, define a
  matroid $\matroidM$ by selecting as bases of the matroid the maximal independent
  subsets of $R$. That these form indeed a matroid follows from the
  Exchange Property in linear algebra. The rank function corresponds
  to span dimension and independence is the same on both sides.

  It is well-known that not all matroids are representable. In fact,
  most are not. 

  Suppose further that $E'\subseteq E$ is a distinguished subset. Let
  $\rho$ be a general linear combination of the elements of
  $R':=\{r_e\}_{e\in E'}$ and consider the quotient set $R/\rho$ in
  the quotient space $R^n/\RR\cdot \rho$. Let $\matroidM_{/\rho}$ be
  the matroid defined by the representation $R/\rho$. Then:
  \begin{enumerate}
  \item The bases $\calB_{\matroidM_{/\rho}}$ of $\matroidM_{/\rho}$
    are precisely the sets $\{e_1,\ldots,e_{\rank(\matroidM)-1}\}$ for
    which there exists $e\in E'$ such that the collection
    $\{r_{e_1},\ldots,r_{e_{r-1}},r_e\}$ is a basis for
    $\matroidM$. In particular, the rank of $\matroidM_{/\rho}$ is
    $\rank(\matroidM)-1$.
  \item The matroid $\matroidM_{/\rho}$ is a quotient of $\matroidM$.
  \end{enumerate}
   It is customary, if $E'=E$, to say that $\matroidM_{\rho}$ is \emph{the
     truncation} of $\matroidM$.

   We do not know whether the matroid of
   momentous 2-forests allows a quotient construction by truncations,
   but it seems unlikely.
\end{rmk}

\subsection{2-forests of $\calG_m$}

Astoundingly, the union of the two matroids $\matroidM^2_{G,\neq}$ and
$\matroidM^2_{G,\mt}$ is also a matroid, as we show next.

\begin{prop}
  The set of 2-forests in the Feynman graph $G$
  that arises as the union of the momentous 2-forests and the
  massively truncated 2-forests forms the set of bases of  a matroid.
\end{prop}
\begin{proof}
  Let $F,F'$ be in $\matroidM^2_{G,\Feynman}$.
  We need to show the validity of the simple Steiner exchange axiom.
  Since $\matroidM^2_{G,\mt}$ and $\matroidM^2_{G,\neq}$ are matroids by
  Lemmas \ref{lem-mass-trunc} and \ref{lem-moment}, it suffices to
  consider the two cases listed below.

  \emph{Case 1: $F$ is momentous and $F'$ is massively truncated.} Let
  $e\in F$ be any edge; then $F\minus\{e\}$ is a 3-forest, with
  components denoted $F_1, F_{2a},F_{2b}$ where $e$ links $F_{2a}$ to
  $F_{2b}$.  Since the set of all 2-forests is in fact a matroid,
  there is at least one edge $g\in F'$ such that $(F\minus
  \{e\})\cup\{g\}$ is a 2-forest. If this is a momentous 2-forest we
  are done with this case. So, in the sequel we assume that no edge of
  $F'$ combines with $(F\minus \{e\})$ to a momentous 2-forest.

  Let $g\in F'$ form a 2-forest $ (F\minus \{e\})\cup\{g\}$.  Then,
  since $F$ itself is momentous and $(F\minus \{e\})\cup\{g\}$
  contains no circuits, such $g$ cannot link $F_{2a}$ to $F_{2b}$ and
  so will connect either $F_1$ to $F_{2a}$, or $F_1$ to
  $F_{2b}$. Depending on the case, the implication would be that the
  external vertices are either completely contained in $F_1\cup
  F_{2a}$ or in $F_{2b}$, or in $F_1\cup F_{2b}$ or in $F_{2a}$. In
  other words, the external vertices are either contained completely
  in $F_1\cup F_{2a}$ or in $F_1\cup F_{2b}$. Without loss of
  generality, let us assume they are all inside $F_1\cup F_{2b}$ and
  so none is in $F_{2a}$. Note that momentousness of $F$ implies that
  some external vertices are in $F_1$ and some in $F_{2b}$.
   In particular then, the egde $g$ from the start of this
  paragraph that creates the non-momentous 2-forest $(F\minus
  \{e\})\cup\{g\}$ connects a vertex of $F_1$ to a vertex of $F_{2b}$.

  It follows that if for no edge $g\in F'$ the set $(F\minus
  \{e\})\cup\{g\}$ is a momentous 2-forest, then all edges of $F'$ are
  either supported on one of $F_1,F_{2a},F_{2b}$, or they must connect
  $F_1$ to $F_{2b}$. That means that all edges of $F'$ are supported
  either on $F_{2a}$, or on $F_1\cup F_{2b}$, implying that the vertex
  sets of $F_{2a}$ and $F_1\cup F_{2b}$ are the same as the vertex
  sets of the two components of $F'$.

  Now recall that $F'$ is massively truncated, and let $f$ be a
  massive edge such that $F'\cup\{f\}$ is a spanning tree. By the
  previous paragraph, $f$ must link a vertex of $F_1\cup F_{2b}$ to
  one of $F_{2a}$. It follows that 
  $(F\minus\{e\})\cup \{g\}$ is massively truncated via $f$.

  \emph{Case 2: $F$ is massively truncated and $F'$ is momentous.} Fix
  an edge $e\in F$, and a massive edge $f$ such that $F\cup\{f\}$ is a
  spanning tree. Then $F\minus \{e\}$ has three components
  $F_1,F_{2a},F_{2b}$ with $e$  linking a vertex from $F_{2a}$
  to one from $F_{2b}$, and
  $f$ linking 
  $F_1$ to either $F_{2a}$ or $F_{2b}$. Without loss of
  generality, assume the latter case.

  Since 2-forests form a matroid, at least one edge $g$ of $F'$ turns
  $F\minus \{e\}$ into a 2-forest.  Suppose all edges $g$ of $F'$ are
  either supported on one of $F_1,F_{2a}$ or $F_{2b}$, or make it
  impossible to certify $(F\minus\{e\})\cup\{g\}$ as massively
  truncated via $f$ (\emph{i.e.}, $(F\minus \{e\})\cup\{g\}\cup\{f\}$
  contains a circuit). Then all edges of $F'$ are either supported on
  one of $\{F_1,F_{2a},F_{2b}\}$, or link $F_1$ to $F_{2b}$. Note that
  therefore an edge $g\in F'$ linking $F_1$ to $F_{2b}$ must exist, as
  else $F'$ should have more than two components.  Since $F'$ has
  exactly two components, these must be supported on $F_1\cup F_{2b}$
  and $F_{2a}$ respectively. Since $F'$ is momentous, $F_{2a}$
  contains some but not all external vertices. Then with the edge
  $g\in F'$ that links a vertex from $F_1$ to one of $F_{2b}$, we find
  that $(F\minus \{e\})\cup \{g\}$ is momentous, finishing the second
  case and the proof.
\end{proof}

\begin{dfn}
We denote the matroid from the previous lemma by
$\matroidM^2_{G,\Feynman}$.
\end{dfn}

\section{Main Theorems}


\subsection{All 2-forests present}

We recall a result from \cite{HT} that will be used in
the proof below.


\begin{thm}
In the massive case,  the semigroup spanned by the lifts of $\Supp
(\calG_m)$ is normal.\qed
\end{thm}\label{thm-massive}

In the massive case, the momenta are inconsequential since the terms
involving $\tilde\Sigma_m\cdot\calU$ alone ensure that the support of
$\calG_m$ is as large as it can possibly be for any mass and any
moment function---keeping in mind Hypothesis \ref{hyp}. We shall prove
here that the conclusion continues to hold as long as every 2-forest
of $G$ contributes to the support of $\calG$; it is immaterial which
terms with squares appear.

For this, recall Equations \eqref{eqn-withE1}, \eqref{eqn-withE2})
and set
\[
\calG_E:=\calU\cdot\tilde\Sigma_E +\calF_0^W.
\]
By the genericity hypothesis on the momenta, all momomials that appear
in $\calG_m$ also appear in $\calG_E$.

\begin{rmk}
An idea that will be used repeatedly is the obvious observation:
\begin{verse}
  (1-forest
complement) $\cup$ (element outside the 1-forest) $=$
(2-forest complement).
\end{verse}

By the (s1I) condition, any given edge $e$ is not a loop, and hence
  contained in a 1-forest $T$.  Removing $e$ one arrives at a 2-forest
  $F=T\minus\{e\}$, those to $\bsx^{E\minus T}$ and to $\bsx^{E\minus
    F}$. If $F$ labels a nonzero term in $\calG_m$ then $A_m$ contains
  two columns whose difference is $\bolde_e$. It follows that $\ZZ
  A_m$ contains $\ZZ^E$, and therefore also $\ZZ\times\ZZ^d$. Thus,
  when all 2-forests are present (and in most other cases), the
  lattice of $A_m$ agrees with the ambient lattice.
\end{rmk}

In the massive case, the supports of $\calG_E$ and $\calG_m$ agree,
and hence the semigroup generated by $(1,\Supp(\calG_E))$ is saturated
in $\ZZ\times \ZZ^d=\ZZ A_m$
by the Helmer--Tellander result. Our strategy will be to show that 
as long as all 2-forests of $G$ contribute to the support of $\calG_m$,
then the semigroup to the lifts of $\Supp(\calG_m)$ can be obtained
from the semigroup to the lifts of $\Supp(\calG_E)$ by intersecting
with suitable half-spaces of $\CC\times \CC^E$. The point is that
halfspaces contain saturated semigroups, and intersections of
saturated semigroups are saturated.

Let us denote by
\[
\mu_e\colon \CC^E\to \CC
\]
the $e$-th coordinate function on $\CC^E$; on
$\CC\times \CC^E$ we include the coordinate function $\mu_0$ on the
first factor into the notation.

  We can now prove the following generalization of \cite[Thm.~1.1,
    part 1]{HT}:
\begin{thm}\label{thm-main}
  Let $G$ be a (s1I) Feynman graph with mass function $m\colon E\to
  \RR_{\geq 0}$ satisfying Hypothesis \ref{hyp}. If
  $\matroidM^2_G=\matroidM^2_{G,\Feynman}$, or equivalently if every 2-forest
  of $\calG$ contributes to $\Supp(\calG)$, then the semigroup $\NN
  A_m$ is saturated and thus the semigroup ring $\KK[\NN A_m]$ is
  normal and Cohen--Macaulay for all fields $\KK$.
\end{thm}
\begin{proof}
  That the second statement follows from the first is contained in
  \cite{Hochster}.

  Comparing the terms in $\calG_m$ and $\calG_E$ in light of our
  assumptions, $\calG_m$ arises from $\calG_E$ by canceling in
  $\tilde\Sigma_E$ all terms that are divided by the square of a
  massless variable, and no others. In other words, 
    the monomials $\bsx^\bolda$ in $\Supp(\calG_m)$ are precisely
    those in $\Supp(\calG_E)$ whose lifted exponent
    $(1,\bolda)$ satisfies $\mu_0((1,\bolda))\geq \mu_e((1,\bolda))$
    for all massless $e\in
    E$.
%

Let $A_E$ denote any matrix whose columns are the lifted support
exponents of $\calG_E$; in particular, we could order $A_E$ in such a
way that $A_m$ becomes a submatrix. For elements $(k,\bolda)$ in $\NN
A_E$ or $\NN A_m$, we call $k=\mu_0((k,\bolda))$ their \emph{degree}.
We have noted above that, as subsets of $\ZZ\times\ZZ^E$,
\[
A_m=A_E\cap \bigcap_{m_e=0}H_e
\]
where
\[
H_e:=\{\alpha\in\RR\times \RR^E\mid \mu_0(\alpha)\geq \mu_e(\alpha)\}
\]
is the half-space on which $\mu_0-\mu_e$ is non-negative.
It follows also that \[
\NN A_m\subseteq (\NN A_E)\cap
\bigcap_{m_e=0}H_e,
\]
and the remainder of the proof is devoted to
showing that this is an equality, which would show that $\NN A_m$ is
the intersection of saturated semigroups, hence saturated itself.

 Take any lattice element $(k,\bolda)$ in the cone $\RR_{\geq
    0}A_m$ of degree $k$.  Since $\NN A_E$ is saturated according to
 Corollary \ref{thm-massive}, one has $(\RR_{\geq 0}A_E)\cap
 (\ZZ\times \ZZ^E)=\NN A_E$.
  Since $(\RR_{\geq 0} A_m)\subseteq (\RR_{\geq 0}A_E)$, one can write
  \begin{eqnarray}\label{eqn-a-decomp}
    (k,\bolda)&=&(1,\bolda_1)+\ldots+(1,\bolda_k)
  \end{eqnarray}
  where each $(1,\bolda_i)$ is a column of $A_E$.

  Note that $(k,\bolda)\in (\RR_{\geq 0}A_m)\subseteq H_e$ for all
  massless $e\in E_0$.  We will show that, given $e\in E_0$, the
  condition $(k,\bolda)\in H_e$ implies that one can rewrite the sum
  \eqref{eqn-a-decomp} in such a way that he following exchange rules hold:
  \begin{itemize}
  \item the new sum only uses summands that are columns of $A_E$;
  \item the number of summands is unchanged;
  \item each summand lies in $H_e$,
  \end{itemize}
  and that, moreover, it can be arranged that
  \begin{itemize}
  \item if all summands were originally in $\bigcap_{e'\in E'} H_{e'}$ for
    some set $E'\subseteq E$,
    then this holds after the rewriting for the larger set $E'\cup\{e\}$.
  \end{itemize}
  Establishing this rewriting forms the main part of the
  proof. Indeed, given such rewriting result, fix a massless edge
  $e\in E_0$. 
  Our exchange rules above allow to change the sum in
  \eqref{eqn-a-decomp} into one where each support vector is in
  $H_e$. Since no exchange operation introduces square terms that were
  not there before, we can treat \eqref{eqn-a-decomp} one $e\in E_0$
  at the time and arrive at a sum as in \eqref{eqn-a-decomp} in which
  every term is in $H_e$ for each $e\in E_0$. But that implies that we
  have written $\bolda$ as a sum of $k$ exponent vectors that appear
  in $\calG_m$, implying that $\NN A_m$ is saturated.

  \bigskip

 Before we engage in the rewriting, note that for $e\in E$, the
 monomials $\bsx^{\bolda_j}$ appearing in $\calG_E=\tilde \Sigma_E\cdot
 \calU+\calF_0$ fall into three categories, depending on whether
 $\mu_e(\bolda_j)$ is $0,1$, or 2. Alternatively, they are classified
 by the value of $(\mu_0-\mu_e)((1,\bolda_j))\in\{-1,0,1\}$.
Those with $\mu_e(\bolda_j)=0$ fall themselves into two  classes:
\begin{enumerate}
\item squarefree monomials without $x_e$ from $\calF_0$ or from
  $\calU\cdot \tilde\Sigma_E$;
\item monomials from $\calU\cdot \tilde\Sigma_E$ that contain a square
  but not $x_e$.
\end{enumerate}


Now suppose that the sum decomposition \eqref{eqn-a-decomp} involves
an element $(1,\bolda_i)$ that is not in the positive real cone of $A_m$ and
therefore satisfies $\mu_e(\bolda_i)=2$ for some (necessarily unique)
$e$ with $m_e=0$. In particular, $\bolda_i$ does then not appear in
$\Supp(\calU)$ and so we will have $|\bolda_i|=r+1$.

Since
$\bolda_i$ is a support vector of $\calG_E$ with
$\mu_e(\bolda_i)=2$, $\bsx^{\bolda_i}$ appears in $\calU\cdot\Sigma_E$
and so
\begin{eqnarray}\label{eqn-ai}
\bsx^{\bolda_i}=\bsx^{E\minus
  T}x_e\qquad \text{with $T\in\calT^1_G$ and $e\not\in T$.}
\end{eqnarray}

Since $(\mu_e-\mu_0)((1,\bolda_i))>0$ but
$(\mu_e-\mu_0)((k,\bolda))\le 0$ there must appear a semigroup element
$(1,\bolda_j)$ in \eqref{eqn-a-decomp} with
$(\mu_e-\mu_0)((1,\bolda_j))<0$; choose one such. It must be of one of
the types (1), (2) or (3) above.

\emph{Case 1:} Suppose $\bolda_j$ is of type (1); then
$\bsx^{\bolda_j}=\bsx^{E\minus F}$ for some 2-forest $F\in\calT^2_G$
with $e\in F$.

The union $T\cup\{e\}$ has exactly one circuit $C$,
$C$ contains $e$, and $F\minus\{e\}$ is  a 3-forest.
Since $C$ is a circuit, $C\minus\{e\}$ has the same span as $C$, and
so $\Span((C\minus \{e\})\cup (F\minus\{e\}))= \Span(C\cup
(F\minus\{e\}))=\Span(C\cup F)$, which contains the 2-forest $F$.
Thus, there is a suitable edge $f\in C\minus\{e\}=C\cap T$ that
combines with the 3-forest $F\minus\{e\}$ to a set of rank greater
than $\rank(F\minus\{e\})$. For such $f$, $(F\minus \{e\})\cup \{f\}$
is therefore a 2-forest. However, so is $T\minus \{f\}$, and so by the
asssuptions of the theorem the monomials
$\bsx^{\bolda'_i}:=\bsx^{E\minus (T\minus \{f\})}$ and
$\bsx^{\bolda'_j}:=\bsx^{E\minus ((F\minus \{e\})\cup \{f\})}$ appear
in $\calF_0$. Moreover, their product is
$\bsx^{\bolda_i'}\bsx^{\bolda_j'}=\bsx^{E\minus T}\bsx^{E\minus
  F}x_e=\bsx^{\bolda_i}\bsx^{\bolda_j}$ and so
$(1,\bolda_i)+(1,\bolda_j)= (1,\bolda'_i)+(1,\bolda'_j)$ in $\NN
A_E$. We can thus replace $\bolda_j$ by $\bolda'_j$ and $\bolda_i$ by
$\bolda'_i$ while preserving \eqref{eqn-a-decomp} as a sum in $\NN
A_E$. Note that the replacement terms have no square terms and so no
new terms with squares in any variable have been introduced while the
overall number of square terms has in fact decreased.

\bigskip

%

\emph{Case 3:} Next consider type (3), where $\bolda_j$ is a support
vector of a term in $\Sigma_E\cdot \calU$ with $\mu_f(\bolda_j)=2$ for
some $f\in E$, while $\mu_e(\bolda_j)=0$. Thus, (we still have
$\bolda_i$ as in \eqref{eqn-ai} and) $\bsx^{\bolda_j}=x_f
\bsx^{E\minus S}$ for some 1-forest $S$ of $G$ that does not
involve $f$ (since else $x_f$ is linear in $x_f \bsx^{E\minus S}$) but does
involve $e$ (so that $x_e$ does not appear in $x_f \bsx^{E\minus S}$).

Then $T\cup\{e\}$ contains a unique circuit $C\ni e$, and the span of
$(C\minus \{e\})\cup(S\minus \{e\})$ contains $\Span(C\cup (S\minus
\{e\}))=\Span(C\cup S)\supseteq \Span(S)=E$. It follows that some
element $g\in (C\minus \{e\})=C\cap T$ different from $e$ turns the
2-forest $S\minus \{e\}$ back into a 1-forest. As removal of $g$ from
$T\cup\{e\}$ breakes the unique circuit $C$ in $T\cup\{e\}$,
$(T\cup\{e\})\minus \{g\}$ is a 1-forest. Then, $(x_e \bsx^{E\minus
  T})\cdot (x_f\bsx^{E\minus S})=(x_e\bsx^{E\minus(T\cup\{e\}\minus
  \{g\})})\cdot (x_f\bsx^{E\minus (S\cup \{g\}\minus \{e\})})$.
In \eqref{eqn-a-decomp}, replace $(1,\bolda_i)+(1,\bolda_j)$ by the
sum of $(1,E\minus(T\minus\{g\}))=(1,\bolda_i+\bolde_g-\bolde_e)$ and
$(1,E\minus (S\cup \{g\}\minus
\{e\}))+(0,\bolde_f)=(1,\bolda_j+\bolde_e-\bolde_g)$.  Both new terms
 are lifts of support vectors of
$\calG_E$, both are in $H_e$, and the only square factor in either one
is $x_f^2$ in the second one, inherited from $\bolda_j$.

\bigskip

This finishes the
rewriting claim, and as explained above proves the theorem.
\end{proof}

It is natural to ask under what conditions we have the equality
$\matroidM^2_G=\matroidM^2_{G,\Feynman}$; we adress this question next.
\begin{dfn}
  A path $v_0,v_1,\ldots,v_t$ of vertices in $G$ (with
  $\{v_i,v_{i+1}\}$ adjacent for all $0\le i<t$) is called
  \emph{massive} if all edges $\{v_i,v_{i+1}\}$ are massive.
\end{dfn}

\begin{thm}\label{thm-when}
  In an (s1I) graph $G$, the equality
  $\matroidM^2_G=\matroidM^2_{G,\Feynman}$ holds if and only if every vertex
  of $G$ permits a massive path to an external vertex of $G$.
\end{thm}
\begin{proof}

  By Hypothesis \ref{hyp}, all momentous 2-forests label a nonzero
  term in $\calG_m$. Thus, assume that the 2-forest $F$ is not
  momentous, so that one of the two components of $F$ contain all
  external vertices.

Then $F$ will cause a nonzero term in $\calG_m$ precisely
if it is a massive truncation. In other words, if and only if there is
a massive edge $e$ such that $F\cup\{e\}$ is a 1-forest.

Since one of the components of $F$ contains all external vertices, the
failure of such a massive edge $e$ to exist implies that the vertices
in the other component of $F$ cannot be linked to $\VExt$ by a massive
path.

Conversely, suppose that some vertex $v$ cannot be linked to $\VExt$
by a massive path. We now delete from $G$ all massive edges and call
the result $G'$. Then $v$ belongs to a connected component $U$ of $G'$
that does not include any external vertex.  Take any 2-forest for $G$ that has
one connected component supported in $U$, and the other 
on $G\minus U$. By our choices, this 2-forest is neither massively
truncated nor momentous and hence does not contribute to $\calG_m$.
\end{proof}

\subsection{The general massless case}

In \cite{HT}, Helmer and Tellander proved that if every vertex of $G$
is an external vertex, then the semigroup $\NN A_m$ is normal for the
mass function that is identically zero. The advantage of the condition
on $\VExt$ is that it places us in a special case of Theorem
\ref{thm-when} above, and  guarantees that 
$\calG_m$ involves a term from every 2-forest,
$\matroidM^2_G=\matroidM^2_{G,\Feynman}$. As it turns out, this condition can
be completely removed: we now use our results from
Section \ref{sec-matroids} to dispose of the general massless case.

We need to review edge-unimodularity and IDP properties of polytopes.
\begin{dfn}
An integer matrix is \emph{unimodular} if all maximal minors are in
the set $\{-1,0,1\}$.

A lattice polytope $P$ is \emph{edge-unimodular} if there is an
  integer 
  unimodular matrix $M$ such that all edges of $P$ are parallel to
  columns of $M$.
\end{dfn}

\begin{dfn}
  A lattice polytope $P\subseteq \ZZ^d$ is said to have \emph{the IDP
    property} or to be \emph{normal} if the intersection $(kP)\cap \ZZ^d$
  agrees with the sum $((k-1)P\cap\ZZ^d)+(P\cap \ZZ^d)$ for all $k\in
  1+\NN$.
\end{dfn}
The benefit of the IDP property to the present context is that it is
equivalent to the equation
\[
\NN((1,P)\cap(\ZZ\times\ZZ^d)) = \RR_{\geq
  0}((1,P))\cap(\ZZ\times \ZZ^d).
\]
In other words, a polytope is IDP
if and only if the semigroup generated by the lattice points in its
lift is saturated in $\ZZ\times \ZZ^d$.

The following result is due to Howard.
\begin{thm}[{\cite[Thm.~4.5]{HowardJA}}]\label{thm-howard}
  Suppose that $A\in\ZZ^{d\times n}$ is a unimodular matrix, and that
  $P$ and $Q$ are lattice polytopes with edges parallel to columns of
  $A$. Then, $(P \cap\ZZ^{d}) + (Q \cap  \ZZ^{d}) = (P + Q) \cap \ZZ^{d}$.
\end{thm}
    In fact, the theorem is stated in a much more constrained context
    (inside a lattice of weights of a Lie algebra) and in a more
    opaque way, but the proof works in the generality stated here
    (which is also the version Howard states in
    \cite[Thm.~1]{HowardOWR}).  As Howard points out, this implies
    that if $P$ is a lattice polytope with edges parallel to the
    columns of a unimodular matrix, then $P$ is IDP and in consequence
    the semigroup generated by the lattice points in the lifted
    polytope $(1,P)$ inside $\ZZ\times \ZZ^d$ is saturated.

\begin{thm}\label{thm-massless}
  Let $G$ be a (s1I) Feynman graph with mass function $m\colon E\to
  \RR_{\geq 0}$ that is identically zero: $m_e=0$ for all $e$. Then
  the semigroup $A_m$ is saturated and thus the semigroup ring
  $\KK[\NN A_m]$ is normal and Cohen--Macaulay for all fields $\KK$.
\end{thm}
\begin{proof}
  The proof follows the one from \cite{HT}, with
  appropriate modifications.

  By our assumptions on $m$, $\calG_m=\calU+\calF_0$. Since the
  momentous 2-forests $\calT^2_{G,\neq}$ form the set of bases of a
  matroid, the support vectors of $\calF_0$ (the complements of the
  elements of $\calT^2_{G,\neq}$ in $E$) are the indicator vectors of
  the bases for the dual matroid $\matroidM^{2,\perp}_{G,\neq}$ on the
  edge set $E$. By \cite{GGMS}, the support polytopes $P^2_{G,\neq}$
  of $\calF_0$ and $P^1_G$ of $\calU$ have their edges within the set
  of vectors $\{\bolde_e-\bolde_{e'}\}_{e,e'\in E}$. The matrix with
  these vectors as columns is unimodular, so the support polytopes of $\calF_0$
  and $\calU$ are edge-unimodular and in particular IDP.

  Since edge directions are invariant under scaling, we have for all
  dilations that $(k\cdot P^2_{G,\neq} +\ell \cdot P^1_G)\cap\ZZ^d = (k\cdot
  P^2_{G,\neq}\cap\ZZ^d)+(\ell \cdot P^1_G \cap\ZZ^d)$. Recall that the Cayley sum of the lattice polytopes $P$ and $Q$ is
  the convex hull of $(\{0\}\times P)\cup(\{1\}\times Q)$ in
  $\RR^{1+d}$. With the IDP
  properties of $P^2_{G,\neq}$ and $P^1_G$ this implies by a theorem of
  Tsuchiya that the Cayley sum of $P^2_{G,\neq}$ and $P^1_G$ has the
  IDP property, \cite[Thm~0.4]{Tsuchiya}.

   Since the entry sums of the vertices of $P^2_{G,\neq}$
  and $P^1_G$ differ by one, an integer  coordinate change shows that the Cayley
  sum of $P^2_{G,\neq}$ and $P^1_G$ can be identified with the convex
  hull of the union
  of $P^2_{G,\neq}$ and $P^1_G$ in $\RR^d$, when embedded into
  $\RR^{1+d}$ by a constant function. It follows that the union of
  $P^2_{G,\neq}$ and $P^1_G$, which is the support polytope of
  $\calG_m$, has the IDP property. So, the semigroup generated by the
  lattice points in the lift $(1,P)$ of this support polytope $P$ is
  saturated.

   Both polytopes $P^2_{G,\neq}$ and $P^1_G$ are matroid polytopes, so
   they have no interior points. They sit in parallel hyperplanes of
   distance one. Thus, the lattice points in their union are precisely
   the lattice points of the two polytopes, which are their vertices.
   Since the vertices are (by definition) support vectors of terms in
   $\calG_m$, the semigroup generated by lifted support vectors is
   saturated.
\end{proof}

\subsection{Approaching the general case}

\begin{prop}
  For all masses and for generic momenta, the support vectors of
  $\calF_m+\calF_0$ are exactly the lattice points inside the support
  polytope of $\calF_m+\calF_0$. In other words, the difference of
  semigroups $\widetilde{\NN A_m}\minus \NN A_m$ has no elements of
  degree 1.
\end{prop}
\begin{proof}
  Suppose $\bolda =\sum\alpha_i\bolda_i$ is a lattice point in the
  support polytope of $\calF_m+\calF_0$ that can be written as a
  linear combination of support vectors of $\calF_m+\calF_m$ with
  $\sum\alpha_i=1$. We need to show that $\bolda$ is a support vector
  itself.

  Each $\bolda_i$ is the support vector of a monomial $\bsx^{E\minus
    T}\cdot x_f$ for some spanning tree $T$ that does not contain the
  edge $f$, or of $\bsx^{E\minus F}$ where $F$ is a momentous
  2-forest. In any event, the entries of $\bolda_i$ are in
  $\{0,1,2\}$. It follows that the same is true for every entry of
  $\bolda$.

  If $\bolda_i$ has a zero entry for some edge $e$, then this must
  also be the case for all $\bolda_i$ with nonzero $\alpha_i$ in the
  linear combination. For such $\bolda_i$, the corresponding tree $T$
  or 2-forest $F$ must contain $e$ (and $e\neq f$ in the tree
  case). Note that spanning trees and 2-forests of $G$ that contain a
  fixed edge $e$ are in bijection with the spanning trees and
  2-forests of the graph $G_{/e}$ obtained from $G$ by contracting the
  edge $e$; the correspondence linking the spanning tree
  (resp.~2-forest) $S\ni e$ of $G$ to the spanning tree
  (resp.~2-forest) $S\minus \{e\}$ of $G_{/e}$. Moreover, $F$ being
  momentous for $G$ is equivalent to $F\minus \{e\}$ being momentou fr
  $G_{/e}$. It follows that we can
  replace $G$ by $G_{/e}$, and $\bolda$ and each $\bolda_i$ by
  $a-\bolde_e$ and $\bolda_i-\bolde_e$ and consider this a computation
  about $G_{/e}$. By induction, the claim is
  already shown for $G_{/e}$, so the case of a zero entry in $\bolda$
  follows.

  If $\bolda$ has an entry 2 for some edge $e$, the same is true for
  every $\bolda_i$ appearing with nonzero coefficient in the linear
  combination. This forces each nonzero term to be of the type
  $\bsx^{E\minus T}\cdot x_e$ with $e\not\in T$, and $e$ must be
  massive. In particular, $e$ cannot be a bridge for $G$ and we can
  polynomially factor $\bsx^{E\minus T}\cdot
  x_e=\bsx^{(E\minus\{e\})\minus T}\cdot x_e^2$. Any $T$ appearing
  here is also a spanning tree for $G_{\minus e}$, the graph obtained
  from $G$ by deleting $e$. Then
  $\sum\alpha_i\bolda_i=(\sum\alpha_i(\bolda_i-\bolde_e)+2\bolda_e$
  and both summands are lattice points. Recall that the set of
  spanning trees in $G_{\minus e}$ is a matroid, and the set of
  complements is the dual matroid. But matroid polytopes have no
  interior points, and so $(\sum\alpha_i(\bolda_i-\bolde_e)$ is one of
  the vertices of the matroid polytope of
  $\matroidM^{1,\perp}_{G_{\minus e}}$. In particular, it must agree
  with one of the terms $\bolda_i-\bolde_e$, and it follows 
  that if $\bolda$ has an entry value of 2, then every lattice point
  in the support polytope is in fact a support vector of
  $\calF_m+\calF_0$.

  We are left to deal with the case where no entry is 0 and no entry
  is 2; thus, all entries are 1. Note that the entry sum of each
  $\bolda_i$, and this also of $\bolda$, is always $|E\minus T|+1$ for
  any spanning tree $T$. But if all entries of $\bolda$ are euqal to
  1, the entry sum is also equal to $E$. So, spanning trees must have
  size 1, which means that (apart for isolated points that make no
  difference to our purposes) $G$ must be a banana graph.

  Suppose $G$ is a banana graph with $m$ massive and $n$ massless
  edges, and let $e_1,\ldots,e_m$ be the massive edges, and suppose
  $\bolda = \sum\alpha_i\bolda_i$ equals $(1,\ldots,1)$. For each
  $\bolda_i$, the massless components of $\bolda_i$ add up to at most
  $n$ since for massless edges no second power can occur in any term
  of $\calG$. But the massless components of $\bolda$ add up to $n$,
  and so each $\bolda_i$ must have the form
  $(c_{i,1},\ldots,c_{i,n},1,\ldots,1)$. Now consider the massive part
  $(-)_m$, the first $m$ components of each vector in the linear
  combination. Since $(\bolda)_m=\sum\alpha_i(\bolda_i)_m$ we have
  reduced the question to the case of a banana graph with
  only massive edges. However, we already know this to be true not
  just for massive banana trees but in fact for all graphs with only
  massive edges.
\end{proof}

\section{Normality vs Cohen--Macaulayness, and Hypothesis \ref{hyp}}

Let $A$ be an integer matrix such that its column span equals the
lattice (free Abelian group) spanned by the unit column vectors;
examples include the matrices $A=A_m$ collecting the support vectors
of $\calG_m$. The semigroup $\NN A$ has an associated
\emph{saturation}, the semigroup $\widetilde{\NN A}$ given by the full
collection of lattice points inside the cone $\RR_{\geq 0}A$. Since
$\NN A\subseteq \widetilde{\NN A}$ and the latter is a semigroup, one
can consider $\widetilde{\NN A}$ as a module over $\NN A$ by
restricting the semigroup operation $\widetilde{\NN A}\times
\widetilde{\NN A}\to \widetilde{\NN A}$ to $\NN A\times
\widetilde{\NN A}$. The resulting semigroup quotient module
$\widetilde{\NN A}/\NN A$ is a measure of the non-saturatedness of
$\NN A$.

On the level of associated semigroup rings, $\tilde
S_A:=\KK[\widetilde{\NN A}]$ is by Hochster's work \cite{Hochster} a
normal Cohen--Macaulay domain, and $S_A:=\KK[\NN A]$ is a subring of
$\tilde S_A$ over which $\tilde S_A$ is a finite integral extension. The
quotient $Q_A:=\KK[\widetilde{\NN A}]/\KK[\NN A]$ is an $S_A$-module.

While $Q_A\neq 0$ is a clear indication that $\NN A$ is not saturated,
it can easily happen that $Q_A\neq 0$ but $S_A$ is
Cohen--Macaulay.

\begin{exa}
  We consider here the massive bubble, whose underlying graph is the
  2-banana graph given as the loopless graph with two vertices (both
  external) and two edges. The only 2-forest has no edge, and there
  are two 1-forests. So $\calU=x_1+x_2$ and $\tilde
  \Sigma_m=1+m_1^2x_1+m_2^2x_2$.  Because of momentum conservation,
  the two external momenta are opposite to one another, and if $p^2$
  denotes the norm at either vertex after Wick rotation then
  $\calF_0=p^2x_1x_2$. So,
  \begin{eqnarray*}
    \calG_m&=&(x_1+x_2)\cdot(1+m_1^2x_1+m_2^2x_2)+p^2x_1x_2\\
    &=&x_1+x_2+m_1^2x_1^2+m_2^2x_2^2+(p^2+m_1^2+m_2^2)x_1x_2
  \end{eqnarray*}
  after Wick rotation. If $p^2+m_1^2+m_2^2=0$,
  $\Supp(\calG_m)=\{{1\choose 0}, {2\choose 0},{0\choose 1}, {0\choose
    2}\}$. The semigroup to the lifted support vectors is not
  saturated since on one hand we have the lattice equation
  \[
  2\begin{pmatrix}1\\1\\1\end{pmatrix}=\begin{pmatrix}1\\2\\0\end{pmatrix}
  + \begin{pmatrix}1\\0\\2\end{pmatrix},
  \]
  and so 2 times $\begin{pmatrix}1\\1\\1\end{pmatrix}$ belongs to the
  semigroup of $A_m$, while on the other hand
  \[
  \begin{pmatrix}1\\1\\1\end{pmatrix}
    = \begin{pmatrix}1\\1\\0\end{pmatrix}
      + \begin{pmatrix}1\\0\\2\end{pmatrix}
        - \begin{pmatrix}1\\0\\1\end{pmatrix}
  \]
  belongs to the lattice spanned by $A_m$.  However, since the toric
  ideal is a hypersurface, it is automatically Cohen--Macaulay.

  The semigroup quotient $Q_A$ consists here of the lattice points
  \[
  \begin{pmatrix}1\\1\\1\end{pmatrix}+\NN \begin{pmatrix}1\\2\\0\end{pmatrix}
  \qquad\text{and}\qquad
  \begin{pmatrix}1\\1\\1\end{pmatrix}+\NN \begin{pmatrix}1\\0\\2\end{pmatrix}.
  \]
\end{exa}

There are certain conditions that $Q_A$
must satisfy for $S_A$ to have the chance of being
Cohen--Macaulay. One of the easiest to decribe concerns the dimension
of the $S_A$-module $\tilde S_A/S_A$, or more precisely the dimensions of its
associated primes. Fortunately, all technical algebraic details can be
expressed in terms of the semigroup quotient $Q_A$. Note the following
easy observation:
\begin{lem}
  The semigroup quotient $Q_A$ is generated over $\NN A$ by the lifted
  support vectors of $\calG_E$ whose coefficients in $\calG_m$ are
  zero, the terms of $\calG_m$ that violate Hypothesis \ref{hyp}. We
  shall denote this set of lattice points by $V_m$.

  If $Q_A$ contains an element $\bolda+\NN A $ such that the elements
  of $\bolda+\NN A\minus \NN A$ are contained in a union of (shifted)
  faces of cone $\RR_{\geq 0} A$ of dimension $\dim(\NN A)-2$ or less,
  then the ring $S_A$ is not Cohen--Macaulay.
\end{lem}
\begin{proof}
  If one adds $V_m$ to the columns of $A_m$ one gets generators for
  the saturation of $\NN A$.

  If $\QQ_A$ contains an element as described in the lemma, then
  $\tilde S_A/\tilde S_A$ has an associated prime of dimension less
  than $\dim(S_A)-1$ and thus has depth less than $\dim(S_A)-1$. By
  standard results on depth, this makes $\depth(S_A)=\dim(S_A)$
  impossible. 
\end{proof}

In order to get a feeling, consider the following example.
\begin{exa}
  Let $G$ be the triple sunset graph on two vertices with three edges
  and no loop, assuming both vertices to be external. Then
  $\calU=x_1x_2+x_2x_3+x_3x_1$,
  $\tilde\Sigma_m=1+m_1^2x_1+m_2^2x_2+m_3^2x_3$. The 2-forests are
  empty, so $\calF_0=p^2x_1x_2x_3$, where $p^2$ is the norm of the
  momentum at either vertex. One computes that in the massive case
  \[
  A_m=\begin{pmatrix}
  1&1&1&1&1&1&1&1&1\\
  1&1&0&2&2&1&0&0&1\\
  1&0&1&1&0&2&2&1&0\\
  0&1&1&0&1&0&1&2&2\end{pmatrix}
  \]
  plus the lift $\bolda_0$ of the support vector of
  $\underbrace{(p^2+m_1^2+m_2^2+m_3^2)}_{:=c_0}x_1x_2x_3$ if the
  coefficient of this term is nonzero.
  
  Let us denote $\bolda_1,\ldots,\bolda_9$ the columns of $A_m$. If
  $c_0$ is nonzero then the semigroup generated by $\Supp(\calG_m)$
  is saturated by Theorem \ref{thm-main}, while otherwise $Q_{A_m}$ is
  generated by $V_m=\bolda_0$.

  In any case, one has the
  identities $\bolda_0+\bolda_1=\bolda_3+\bolda_4\in\NN A_m$ and
  $\bolda_0+\bolda_4=\bolda_5+\bolda_6\in\NN A_m$. It follows from
  symmetry that $\bolda_0+\bolda_i\in\NN A_m$ for $1\le i\le 9$ and so
  $Q_{A_m}$ is the singleton $\{\bolda_0\}$. Equivalently, the
  $S_A$-module $\tilde S_A/S_A$ is a 1-dimensional vector space in
  multi-degree $(1,1,1,1)$. 

  Application of the long Euler--Koszul homology functor from
  \cite{MMW} to the short exact sequence $S_A\to \tilde S_A\to \tilde
  S_A/S_A$ now implies that the GKZ-system attached to $A_m$ with
  parameter $\beta$ has a larger solution space (namely, of dimension
  $v+9-1$) than all other GKZ-systems attached to $A_m$ (whose rank is
  always the volume $v$ of the comvex hull of $A_m$). In particular,
  $S_{A_m}$ is not Cohen--Macaulay.

  An alternative way using commutative algebra is  to
  observe that $\tilde S_{A_m}/S_{A_m}$ being a finite dimensional
  vector space (that is, a zero-dimensional module) means that as
  $S_{A_m}$-module it must have depth zero, which then forces
  $S_{A_m}$ to have depth one. But as the dimension of $S_{A_m}$ is
  equal to the dimension of the lattice spanned by $A_m$ (namely, 4),
  $S_{A_m}$ is far from satisfying the equality
  $\dim(S_{A_m})=\depth(S_{A_m})$ that determines Cohen--Macaulayness.
\end{exa}

In general, if $\tilde S_A/S_A$ contains a submodule of dimension
$k<\dim(S_A)-1$ then the depth of $S_A$ cannot exceed $k+1$ and thus
$S_A$ cannot be Cohen--Macaulay.
\begin{exa}
  Suppose $S_A$ is the $\KK$-algebra inside the polynomial ring
  $\KK[x,y]$ generated by the monomials $x^3y,x^2,xy^2,y$, $\KK$ a
  field. Then $S_A$ contains all monomials of $y$-degree 2 or more,
  all powers of $x^2$, and all monomials $x^ty$ except for $t=1$. Then
  $Q_A=\{(1,1)\}\cup\{2t+1,0\}_{t\in\NN}$. It would be reasonable to
  say that $Q_A$ is 1-dimensional since it spreads out infinitely far
  along the line $(*,0)$. However, $Q_A$ contains a submodule of
  dimension zero, 
  generated by the monomial $xy$, since $x^2y$ and $xy^2$ are in
  $S_A$. It follows that $S_A$ has depth zero and is not Cohen--Macaulay.

\end{exa}
It seems very likely that Cohen--Macaulayness of $S_A$ fails most of
the time 
that Hypothesis \ref{hyp}.(2) fails. 
  
\section{List of symbols}

\begin{itemize}
  \item $(G,m,\VExt)$ a Feynman graph with edge set $E$, mass function
    $m\colon E\to\RR$ and external vertices $\VExt$.
  \item $E_m,E_0\subseteq E$ the sets of massive and of massless edges.
  \item $\calT^i_G$ the set of $i$-forests of $G$.
  \item $\matroidM^i_G$ the matroid whose bases are the $i$-forests of
    $G$.
  \item $\matroidM^2_{G,\neq}$ the matroid whose bases are the
    momentous 2-forests of $G$.
  \item $\matroidM^2_{G,\mt}$ the matroid whose bases are the massively
    truncated 2-forests of $G$.
  \item $\matroidM^2_{G,\Feynman}$ the matroid whose bases label the
    square-free terms
    in $\calG_m$.
  \item $\calU$ the first Symanzik polynomial.
  \item $\calF_0^W$ the sum over $\matroidM^2_{G,\neq}$
    weighted with their Wick rotated moments.
  \item $\tilde \Sigma_m=1+\Sigma_m=1+\sum m_e^2x_e$.
  \item $\calG_m=\tilde \Sigma_m\cdot \calU+\calF_0^W$ the Feynman
    integrand. 
  \item $\tilde \Sigma_E=1+\Sigma_E=\Sigma_m+\sum_{m_e=0}x_e$.
  \item $\calG_E=\calU\cdot\tilde\Sigma_E +\calF_0^W$.
  \item $P_m$ the support polytope of $\calG_m$.
\end{itemize}

\bibliographystyle{amsalpha}
\bibliography{normal.bib}
\end{document}